\documentclass[12pt,a4paper,oneside,reqno,notitlepage]{amsart}
\usepackage{amssymb}
\textwidth
150mm

\newtheorem{theorem}{Theorem}
\theoremstyle{plain}

\newtheorem{definition}{Definition}
\newtheorem{example}{Example}

\newtheorem{lemma}{Lemma}

\newtheorem{remark}{Remark}

\numberwithin{equation}{section}
\usepackage{graphicx}
\usepackage{subfigure}

\begin{document}

\baselineskip 8mm
\parindent 9mm

\title[]
{$NP/CLP$ Equivalence: A Phenomenon Hidden Among Sparsity Models for Information Processing}

\author[author1]{Peng Jigen$^\dag$}
\author[author2]{Yue Shigang$^\ddag$}
\author[author3]{Li Haiyang$^\S$}

\thanks{$\dag$ School of Mathematics and Statistics, Xi'an
Jiaotong University, Xi'an 710049, China; and Beijing Center for Mathematics and Information Interdisciplinary Sciences, Beijing 10048, China; Email: jgpeng@mail.xjtu.edu.cn}

\thanks{$\ddag$ Corresponding author. School of Computer Science, University of Lincoln, Lin-
coln LN6 7TS, UK; Email: syue@lincoln.ac.uk; yue.lincoln@gmail.com}

\thanks{$\S$ School of Science, Xi¡¯an polytechnic University, Xi¡¯an, 710048, China; Email: fplihaiyang@126.com}

\begin{abstract}
It is proved in this paper that to every underdetermined linear system $Ax=b$ there corresponds a constant $p(A,b)>0$ such that every solution to the $l_p$-norm minimization problem also solves the $l_0$-norm minimization problem whenever $0<p<p(A,b)$. This phenomenon is named $NP/CLP$ equivalence.
\end{abstract}
\maketitle

\section{\textbf{Introduction}}

In sparse information processing, the following minimization is commonly employed to model basic sparse problems such as sparse representation and sparse recovery,
\begin{equation}\label{equ1}
(P_0)\hspace{15mm} \min\limits_{x}\|x\|_0\ \ {\rm subject\  to}\ \ Ax=b
\end{equation}
where $A$ is a real matrix of order $m\times n$ with $m<n$, $b$ is a nonzero real vector of $m$-dimension, and $\|x\|_0$ is the so-called $l_0$-norm of real vector $x$, which counts the number of the non-zero entries in $x$ \cite{bruckstein,elad,theodoridis}. Unfortunately, although the $l_0$-norm provides a very simple and essentially grasped notion of sparsity, the optimization problem $(P_0)$ is actually NP-Hard and thus quite intractable in general, due to the discrete and discontinuous nature of the $l_0$-norm. Therefore, many substitution models for $(P_0)$ have been proposed through relaxing $l_0$-norm as the evaluation function of desirability of the would-be solution to $(P_0)$ (see, for example, \cite{candes,donoho2,gorodnitsky,herzet,theodoridis}, and references therein). With the following relationship
\begin{equation}\label{equ2}
\|x\|_0=\lim\limits_{p\rightarrow 0^+}\sum^n_{i=1}|x_i|^p=\lim\limits_{p\rightarrow 0^+}\|x\|_p^p, \ \forall x=(x_1,x_2,\cdots, x_n)^T,
\end{equation}
the following minimizations seem to be among the most natural choices,
\begin{equation}\label{equ3}
 (P_p)\hspace{15mm} \min\limits_{x}\|x\|_p^p\ \ {\rm  subject\  to}\ \ Ax=b
\end{equation}
where $0<p\leq 1$. Indeed, the above optimization models, particularly for the special case when $p=1$, have gained its popularity in the literature (see, e.g., \cite{chartrand,davies,foucart,lai,saab,sun,wang,xu}), since the remarkable work done by Donoho and Huo \cite{donoho1} and Candes and Tao\cite{candes2} for $p=1$ and the initial work by Gribnoval and Nielsen \cite{gribonval} for $0<p<1$. However, with respect to these choices, a central problem is to what extent the minimizations $(P_p)$ can achieve the same results as the initial  minimization $(P_0)$. A lot of excellent theoretical work (see,e.g.,\cite{candes,donoho1,donoho3,gribonval}), together with some empirical evidence (see, e.g.,\cite{chen}), have shown that the $l_1$-norm minimization $(P_1)$ can really make exact recovery provided some conditions such as the restricted isometric property (RIP) are assumed. As an original notion RIP has received much attention, and has already been tailored to the more general case when $0<p<1$ (see, e.g.,\cite{chartrand,davies,herzet}). Among those publications mentioned, we would like to especially refer to the excellent work done by Donoho and Tanner in \cite{donoho2}, there they expose such an amazing phenomenon by means of convex geometry that for any real matrix $A$, whenever the nonnegative solution to $(P_0)$ is sufficiently sparse, it is also unique solution to $(P_1)$. That is, there exists a certain equivalence between $(P_0)$ and $(P_1)$. As the former is discrete and so NP-Hard and while the latter is continuous and equivalent to a linear programming (LP), this phenomenon was called $NP/LP$ equivalence therein. It is worthwhile to note that $(P_0)$ and $(P_1)$ are just the extremes of $(P_p)$ with respect to $p$ in the interval $(0,1)$, and that the relationship (\ref{equ3}) together with its geometric illustration in Figure 1 appears to indicate the more aggressive tendency of $(P_p)$ to drive its solution to be sparse as $p$ decreases. So, it is natural to believe that there exists more general equivalence between $(P_0)$ and $(P_p)$ for $p\in (0,1)$. Motivated by this, we in this paper aim to expose the equivalence.
\begin{figure}[h!]
  \centering
   \includegraphics[height=0.38\textwidth,width=0.48\textwidth]{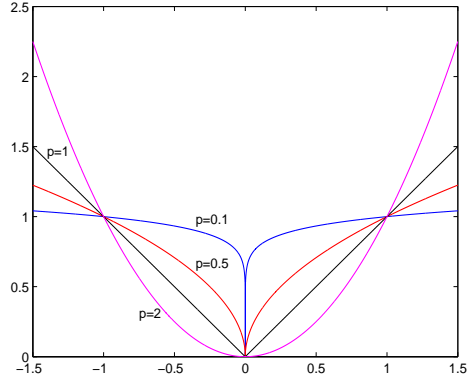}
     \caption{{\small The behaviour of $l_p$-norm for various values of $p$. See,e.g.\cite[p.38]{bruckstein}}}
    \label{fig1}
\end{figure}

The remainder of this paper is organized as follows. In Section 2 we first set up a decomposition of the background space $\mathbb{R}^n$ with respect to the system $Ax=b$, and then derive constructions and locations of solutions to $(P_p)$ based on the classical Bauer maximum principle. Section 3 is devoted to proving the main theorem, which establishes an equivalence between $(P_p)$ and $(P_0)$, named $NP/CLP$ equivalence. Finally we conclude this paper in Section 4.

For convenience, in this paper we denote by $\mathbb{R}^n$ the $n$-dimension real space, and for a vector $x\in\mathbb{R}^n$ by $x_i$ its $i^{th}$ component and by $|x|$ its module vector (i.e., $|x|=(|x_1|,|x_2|,\cdots, |x_n|)^T$). We also use $\mathbb{R}^n_+$ to represent the positive cone $\{x\in\mathbb{R}^n: x_i\geq 0,i=1,2,\cdots,n\}$.

\section{\textbf{Preliminaries: Constructions and Locations of Solutions to $(P_p)$}}

As a preliminary section, this section is devoted to explore how to construct the solutions to the problems $(P_p)$ and where the solutions locate. For the minimization problems $(P_p)$ are subject to underdetermined linear systems, we first are going to analyze the constructions of solutions to linear equation $Ax=b$, where $A$ be an $m\times n$ real matrix with $m<n$ and $b$ an $n$-dimension real vector. Without loss of generality, we assume throughout this paper that $A$ has full row rank (otherwise, we make row transformations simultaneously in both sides of the equation $Ax=b$, resulting in an equivalent equation $A_1x=\hat{b}$ with $A_1$ being of full row rank). So, by the nature of geometric aspect of $A$, we can get a construction decomposition of solutions to $Ax=b$ as below.
\begin{lemma}\label{lemma1}
Denote by $N(A)$ the null space of $A$, and $R(A^T)$ the range of $A^T$ of the transposition of $A$. Then, there holds the following space decomposition
\begin{equation}\label{equ4}
\mathbb{R}^n=N(A)\oplus R(A^T),
\end{equation}
that is, for every $x\in \mathbb{R}^n$ there uniquely exist $x_N\in N(A)$ and $x_R\in R(A^T)$ such that $x=x_N+x_R$. Therefore, every solution to $Ax=b$ can be  explicitly expressed as $x=x_N+A^T(AA^T)^{-1}b$ with $x_N\in N(A)$.
\end{lemma}
\begin{proof}
It is clear to see that $R(A^T)$ as a linear subspace of $\mathbb{R}^n$ is closed since $A$ is of full row rank. So, in order to prove the decomposition formula  (\ref{equ4}) it suffices to show that $N(A)$ is just the orthogonal complement of $R(A^T)$, i.e., $N(A)=R(A^T)^{\perp}$. Obviously, $N(A)\subseteq R(A^T)^{\perp}$. For the converse containing, let $z\in R(A^T)^{\perp}$. By definition, we have that $\langle Az,y\rangle=\langle z,A^Ty\rangle=0$ for all $y\in\mathbb{R}^m$, which implies that $Az=0$, namely $z\in N(A)$. Thus, $R(A^T)^{\perp}\subset N(A)$ by the arbitrariness of $z\in R(A^T)^{\perp}$.

Let $x$ be a solution to $Ax=b$. Then, by (\ref{equ4}) there uniquely correspond $x_N\in N(A)$ and $y\in \mathbb{R}^m$ such that $x=x_N+A^Ty$. By definitions, $b=Ax=AA^Ty$ and hence $x=x_N+A^Ty=x_N+A^T(AA^T)^{-1}b$, as claimed.
\end{proof}

In the above proof, we have used the notation $\langle x, y\rangle$ to represent the inner product of real vectors $x$ and $y$. For sake of convenience, we will adopt it hereafter.

\begin{remark}\label{re1}
The space decomposition (\ref{equ4}) implies that each $x_N\in N(A)$ can be written into the form $x_N=h-A^T(AA^T)^{-1}Ah$ with $h=x_N+A^Tb$, and so that the null of $A$ has such a parameterized expression as $N(A)=\{h-A^T(AA^T)^{-1}Ah: h\in\mathbb{R}^n\}$. As a result, the constrained minimizations $(P_p)$ for all  $p\geq 0$ can be equivalently transformed into the global minimizations as follows
\begin{equation}\label{equ5}
\min\limits_{h\in \mathbb{R}^n}\|h+A^T(AA^T)^{-1}(b-Ah)\|_p^p.
\end{equation}
On the other hand, it follows by Lemma 1 that for all $p>0$ the optimal values to $(P_p)$ are upper bounded by $\|A^T(AA^T)^{-1}b\|^p_p$, i.e.,$\|x\|^p_p\leq\|A^T(AA^T)^{-1}b\|^p_p$ for every solution $x$ to $(P_p)$. This means that the $(P_p)$s are actually constrained within, addition to $Ax=b$, a bounded subset, for example, the $l^\infty$ ball $B_\infty(r):=\{x\in\mathbb{R}^n: |x_i|\leq r, i=1,2,\cdots, n\}$ with $r=n\cdot\sup\limits_{1\leq i\leq n}|(A^T(AA^T)^{-1}b)_i|$. That is, $(P_p)$ with $p>0$ is equivalent to the following minimization problem
\begin{equation}\label{equ3-add}
 (P^\prime_p)\hspace{15mm} \min\limits_{Ax=b, x\in B_\infty(r)}\|x\|_p^p.
\end{equation}
It is valuable to note that the optimal problem (\ref{equ3-add}) is subject to a special constraint set, which is called typical polytope in term of the following definition.
\end{remark}

\begin{definition}\cite{goemans}\label{def1} A polyhedron $G$ in $\mathbb{R}^n$ is a subset as the intersection of finitely many halfspacs, where halfspace refers to a set of the form $\{x\in\mathbb{R}^n: \langle h,x\rangle\leq \gamma\}$ for some vector $h\in\mathbb{R}^n$ and real number $\gamma\in\mathbb{R}$. Moreover, a polyhedron is called polytope if it is bounded.
\end{definition}

By definition, a polyhedron can be compactly expressed as $G=\{x\in \mathbb{R}^n: Hx\leq g\}$ for some matrix $H$ and vector $g$ of a certain dimension, where $Hx\leq g$ means that the corresponding inequalities hold for all scale components (i.e., $\langle H_i, x\rangle\leq g_i$, where $H_i$ is the $i^{th}$ row of $H$. Below we always adopt this notation). Moreover, it is clear that a polytope is closed and convex, and the intersection of any finitely many polyhedrons is a polytope as long as some of those polyhedrons is bounded. Obviously, the $l_\infty$ ball $B_\infty(r)$ (see Remark 1), as well as the set of solutions to $Ax=b$, is a polyhedron, and so its intersection $\Omega=\{x\in\mathbb{R}^n: Ax=b, |x_i|\leq r, i=1,2,\cdots,n\}$ is a polytope. Indeed, if let $e_i$ denote the $i^{th}$ normal base vector of $\mathbb{R}^n$ and define a $2(n+m)\times n$ matrix $H$ as
\begin{equation}\label{equ6}
H=(e_1, -e_1; e_2,-e_2;\cdots; e_n, -e_n; A^T, -A^T)^T
\end{equation}
and a $2(n+m)$-dimension real vector $g$ as
\begin{equation}\label{equ7}
g=(r,r;r,r;\cdots;r,r; b^T,-b^T)^T,
\end{equation}
then the intersection set $\Omega$ is compactly expressed by $\Omega=\{x\in\mathbb{R}^n: Hx\leq g\}$, a standard polytope.

According to convex polytope theory \cite{grunbaum}, we know that polytope possesses such an important characteristic that it can be generated by convexifying a set of finite number of points, which are no other than extreme points. In convexity analysis, extreme points always play vital roles. For example, the optimal solution to a linear programming, whose constraint set is generally a polytope (or polyhedron), can be achieved at some extreme point (see, e.g.,\cite{grunbaum,simon} and references therein). Below we first recall the definition of extreme point.

\begin{definition}\cite[Chapt.8]{simon}\label{def2}
Let $\Omega$ be a set of a vector space, and let $x^*\in\Omega$. Then, $x^*$ is called extreme point of $\Omega$ if it does not lie in the interior of any line segment entirely contained in $\Omega$ $($i.e., $x^*$ necessarily coincides with $x_1$ or $x_2$ whenever $x^*=\alpha x_1+(1-\alpha)x_2$ with $x_1, x_2\in\Omega$ and $\alpha\in [0,1])$.
\end{definition}

The famous Klein-Milman theorem \cite{aliprantis} tells that every compact convex subset of a locally convex topological vector space necessarily possesses extreme point(s) and it is just the closed convex hull of those extreme points. As a particular case, Minkowski-Caratheodory theorem \cite[Chap.8, p.126]{simon} states that every point from a compact convex subset of $\mathbb{R}^n$ is a convex combination of at most $n +1$ extreme points. It is well-known that a polytope possesses at most finite number of extreme points (indeed, if $G=\{x\in\mathbb{R}^n: Hx\leq g\}$ is a polytope, then it has at most $2^m$ extreme poins, where $m$ is the dimension of $g$. See, e.g.,\cite{murty}). We previously mentioned that a linear programming can attain its optimal value at some extreme point of the constraint set. In fact, H. Bauer \cite{bauer} had proved this assertion early in 1960 for a general convex maximization problem with compact constraint set. We state it as a lemma below.

\begin{lemma}(Bauer's Maximum Principle \cite[Chapt.7, p.298]{aliprantis})\label{lemma2}
If $G$ is a compact convex subset of locally convex Hausdorff vector space, then every continuous convex functional on $G$ has a maximizer that is an extreme point of $G
$.
\end{lemma}

Recall that a functional $f$ on a convex set $G$ is said to be concave if there holds the inequality $\alpha f(x)+(1-\alpha)f(y)\leq f(\alpha x+(1-\alpha y))$ for all $x,y\in G$ and $\alpha\in (0,1)$. In this terminology, Lemma 2 means that every continuous concave real functional defined on a compact convex subset of locally convex Hausdorff vector space achieves its minimum value at some extreme points. By definitions, a linear function is convex but also can be considered as being concave. This is just why a linear programming problem always reaches its optimal value at some extreme point, whether it is of minimization problem or maximization one. Consider the minimization problem $(P_1)$ under the additional nonnegative constrains $x\geq 0$, or equivalently the following linear programming
\begin{equation}\label{equ8}
(LP)\hspace{15mm} \min\limits_{x} \langle \mathbf{1}, x\rangle \ \ {\rm subject\  to}\ \ Ax=b, x\geq 0,
\end{equation}
where $\mathbf{1}\in \mathbb{R}^n$ is the vector whose components are all one. Due to the nonnegativity assumption on the variables $x$, it is easy to show that the constraint set is actually bounded and hence a polytope. So by Lemma 2 and Minkowski-Caratheodory theorem it is clear that the minimization problem $(P_1)$ is solvable by searching for the extreme points of the constraint set. However, the extreme point set possesses up to $2^m$ members, so such searching throughout all extreme points is a overwhelming task for large $m$, which may be an implicit reason why $(P_0)$ is equivalent to $(P_1)$ in some cases.

Before close this section, we present a further remark on Lemma 2, which will be useful in proof of our main result.

\begin{remark}\label{re3}
Recall that a concave function $f$ is said to be strict if the equality in $\alpha f(x)+(1-\alpha)f(y)\leq f(\alpha x+(1-\alpha y))$ is available only for $x=y$. In this meaning, it is easy to check that every minimizer of a strictly concave function on a compact convex set $G$ necessarily exists at a certain extreme point of $G$. Obviously, the function $f_p(x)=\|x\|^p_p$ with $0<p<1$ is strictly concave in the positive cone $\mathbb{R}^n_+$.
\end{remark}

\section{\textbf{Main Result: Equivalence Between Minimizations $(P_0)$ and $(P_p)$}}

With the preparations given in the last section, we in this section establish an equivalence between $(P_0)$ and $(P_p)$. To this end, we first turn to the minimization problems $(P_p)$ for $0<p\leq 1$ and present the following lemma.

\begin{lemma}\label{le3}
Let $A$ be an $m\times n$ real matrix of full row rank, and $b\in\mathbb{R}^n$. Given $r>0$ and define a subset of $\mathbb{R}^n$ as
\begin{eqnarray}\label{equ9}
&&G(r)=\{z\in\mathbb{R}^n: \forall i=1,2,\cdots,n, 0\leq z_i\leq r, \text{and}\nonumber
\\ &&\hspace{15mm}\text{there is a solution}\  x\  \text{to } Ax=b \ \text{such that}\ |x|\leq z\},
\end{eqnarray}
where $|x|$ stands for the module vector of $x$. Then, $G(r)$ is a polytope in $\mathbb{R}^n$.
\end{lemma}
\begin{proof}
Obviously, $G(r)$ is bounded and closed. Below we prove that $G(r)$ is a polyhedron. According to Lemma 1 (combined with Remark 1), we know that $x\in\mathbb{R}^n$ solves the system $Ax=b$ if and only if it bears the form $x=h-A^T(AA^T)^{-1}Ah+A^T(AA^T)^{-1}b$ for some $h\in\mathbb{R}^n$. Denote by $\Lambda$ the set of those vectors $(z^T,h^T)^T$ of $\mathbb{R}^n\times\mathbb{R}^n$ satisfying the following two inequalities
\begin{equation}\label{equ10}
-z-(I-A^T(AA^T)^{-1}A)h\leq A^T(AA^T)^{-1}b,
\end{equation}
and
\begin{equation}\label{equ11}
-z+(I-A^T(AA^T)^{-1}A)h\leq -A^T(AA^T)^{-1}b.
\end{equation}
Obviously, $\Lambda$ is a polyhedron of the product space $\mathbb{R}^n\times\mathbb{R}^n$. Let $P$ be the projection from $\mathbb{R}^n\times\mathbb{R}^n$ to the first part, i.e., $P(z,y)=z$. Then, by the projection property of polyhedron \cite{goemans} we know that $P(\Lambda)$ is also a polyhedron of $\mathbb{R}^n$. Therefore, to close the proof, we only need to show that $G(r)=P(\Lambda)\cap B^+_\infty(r)$, where $B^+_\infty(r)$ stands for the subset of nonnegative elements of $B_\infty(r)$.

In light of the above discussion, it is clear that $G(r)\subseteq P(\Lambda)\cap B^+_\infty(r)$. For the converse containing relation, let $z\in P(\Lambda)\cap B_\infty^+(r)$. Then, $0\leq z_i\leq r$ for all $i=1,2,\cdots,n$, and there corresponds a $h\in \mathbb{R}^n$ such that $(z^T,h^T)^T\in \Lambda$, that is, the pair $(z, h)$ satisfies both the inequalities (\ref{equ10}) and (\ref{equ11}). Let $x=h+A^T(AA^T)^{-1}(b-Ah)$. Then, it is a routine work to check that $x$ solves $Ax=b$, and moreover it follows from the inequalities (\ref{equ10}) and (\ref{equ11}) that $|x|\leq z$. So $z\in G(r)$. The proof is therefore completed.
\end{proof}

Figure 2 displays basic shapes of $G(r)$ in the plane $\mathbb{R}^2$, except several degenerate cases (e.g., the segment between these two points $(0,r)^T$ and $(r,r)^T$).

\begin{figure}[h!]
 \centering
  \includegraphics[height=0.25\textwidth,width=0.8\textwidth]{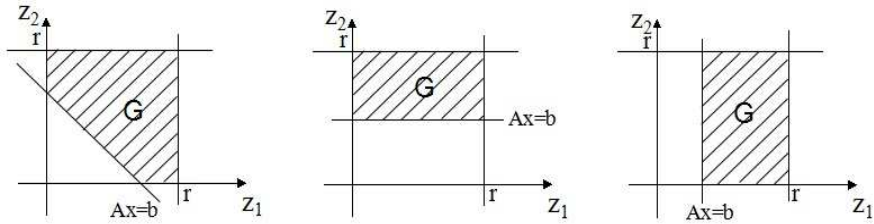}\\
 \caption{Shapes of the polytope $G(r)$ defined as in (\ref{equ9}) in the plane.}\label{fig2}
\end{figure}

In order to prove the following theorem, the main result of this paper, we make a remark on the solutions to $(P_0)$. Due to the integer-value virture of $l_0$-norm, the optimal value of $(P_0)$ is certainly achieved in a bounded set. That is, there exists a constant $r_0>0$ such that
\begin{equation}\label{equ11-add}
\min\limits_{Ax=b}\|x\|_0=\min\limits_{Ax=b, x\in B_\infty(r_0)}\|x\|_0.
\end{equation}
(Geometrically, it appears to be true that the optimal solutions to $(P_0)$ locate at those intersections of the plane $Ax=b$ with coordinate axises or coordinate planes).

\begin{theorem}\label{the1}
There exists a constant $p(A,b)>0$ such that, whenever $0<p<p(A,b)$, every solution to $(P_p)$ also solves $(P_0)$.
\end{theorem}
\begin{proof}
Let $G(r_1)$ be defined as in (\ref{equ9}) with $r_1=n\cdot\sup\limits_{1\leq i\leq n}|(A(AA^T)^{-1}b)_i|$. Then, by Lemma 3 we know that $G(r_1)$ is a polytope and hence has finite number of extreme points. Denote by $E(G(r_1))$ the set of extreme points of $G(r_1)$, and define a constant $r_m(A,b)$ as follows
\begin{equation}\label{equ12}
r_m(A,b)=\min\limits_{z\in E(G(r_1)),z_i\not=0, 1\leq i\leq n}z_i.
\end{equation}
Clearly, the defined constant $r_m(A,b)$ is finite and positive (i.e., $0<r_m(A,b)<\infty$), due to the finiteness of $E(G(r_1))$. (Here we draw attention to that $r_m(A,b)$ only depends on $A$ and $b$. However, for simplicity we may draw off $A,b$ from $r_m(A,b)$ in the following discussion).

Let $r_0$ be given as in (\ref{equ11-add}), and $r=\max\{r_0, r_1\}$. Then, we have $G(r_1)\subset G(r)$, and similar to (\ref{equ3-add}) we have
\begin{equation}\label{equ12-add}
\min\limits_{Ax=b}\|x\|_p^p=\min\limits_{Ax=b, x\in B_\infty(r_1)}\|x\|_p^p=\min\limits_{Ax=b, x\in B_\infty(r)}\|x\|_p^p.
\end{equation}
and
\begin{equation}\label{equ121-add}
\min\limits_{Ax=b}\|x\|_0=\min\limits_{Ax=b, x\in B_\infty(r_0)}\|x\|_0=\min\limits_{Ax=b, x\in B_\infty(r)}\|x\|_0.
\end{equation}

Any given a solution $x_p$ to $(P_p)$. By Remark 1 we have $x_p\in B_\infty(r_1)$. Now, let $z_p=|x_p|$. Then, by Lemma 3 it is clear to see that $z_p\in G(r_1)$ and solves the following minimization problem
\begin{equation}\label{equ13}
(P_p^+)\hspace{15mm}\min\limits_{z\in G(r_1)}\|z\|^p_p=\sum^n_{i=1} z_i^p.
\end{equation}
Moreover, since $f(z)=\|z\|^p_p$ is strictly concave in $G(r_1)$, it follows by Lemma 2 and its Remark 2 that $z_p\in E(G(r_1))$. Hence, by noticing the fact that the mapping $p\rightarrow x^p$ is decreasing for $x\in (0,1)$ and while increasing for $x>1$, we have
\begin{eqnarray*}\label{equ13}
\|x_p\|_0&=&\|z_p\|_0=\lim\limits_{q\downarrow 0}\big[z^q_{p1}+z^q_{p2}+\cdots +z^q_{pn}\big]\\
&=&\lim\limits_{q\downarrow 0}\bigg[\bigg(\frac{z_{p1}}{r_m}\bigg)^q+\bigg(\frac{z_{p2}}{r_m}\bigg)^q+\cdots +\bigg(\frac{z_{pn}}{r_m}\bigg)^q\bigg]\\
&\leq&\bigg(\frac{z_{p1}}{r_m}\bigg)^p+\bigg(\frac{z_{p2}}{r_m}\bigg)^p+\cdots +\bigg(\frac{z_{pn}}{r_m}\bigg)^p\\
&=&r_m^{-p}\min\limits_{z\in G(r_1)}\|z\|_p^p=r_m^{-p}\min\limits_{Ax=b,x\in B_\infty(r_1)}\|x\|_p^p\\
&=&r_m^{-p}\min\limits_{Ax=b,x\in B_\infty(r)}\|x\|_p^p=\bigg(\frac{r}{r_m}\bigg)^p\min\limits_{Ax=b,x\in B_\infty(r)}\|r^{-1}x\|_p^p\\
&\leq&\bigg(\frac{r}{r_m}\bigg)^p\min\limits_{Ax=b,x\in B_\infty(r)}\|r^{-1}x\|_0\\
&=&\bigg(\frac{r}{r_m}\bigg)^p\min\limits_{Ax=b,x\in B_\infty(r)}\|x\|_0=\bigg(\frac{r}{r_m}\bigg)^p\min\limits_{Ax=b}\|x\|_0,
\end{eqnarray*}
where the last equality follows from (\ref{equ121-add}). Because $\|x_p\|_0$ is an integer number, from the above inequality it follows that  $\|x_p\|_0=\min\limits_{Ax=b}\|x\|_0$ (that is, $x_p$ solves $(P_0)$) when
\begin{equation}\label{equ14}
\bigg(\frac{r}{r_m}\bigg)^p\min\limits_{Ax=b}\|x\|_0<\min\limits_{Ax=b}\|x\|_0+1
\end{equation}
Obviously, the above inequality is true whenever
\begin{equation}\label{equ15}
p<\frac{\ln\bigg(\min\limits_{Ax=b}\|x\|_0+1\bigg)-\ln\bigg(\min\limits_{Ax=b}\|x\|_0\bigg)}{\ln r-\ln r_m}.
\end{equation}
Therefore, denoting by $p(A,b)$ the right side of the above inequality, we conclude that when $0<p<p(A,b)$, every solution $x_p$ to $(P_p)$ also solves $(P_0)$. The proof is thus completed.
\end{proof}

\begin{remark}\label{re3}
From the proof it is clear to see that the parameters $r_0$ and $r_1$ are just used to bound the ranges of solutions to $(P_0)$ and $(P_1)$ respectively. So, to obtain a better $p(A,b)$ we can replace the number $r$ in the inequality (\ref{equ15}) with another number that bounds the solutions to $(P_p)$ for all $0\leq p\leq 1$.
\end{remark}

As is well known, $(P_0)$ is combinatorial and NP-hard in general, and while $(P_p)$ for $p>0$ is continuous and perhaps is polynomially computable. In this sense, the significance of the theorem lies in that it really bridges the gap between a combinatorial problem and a continuous one. To highlight the NP nature of $(P_0)$ and the continuity feature of $(P_p)$, we name the phenomenon stated by Theorem 1 ``$NP/CLP$ equivalence". Correspondingly, we call the maximal $p(A,b)$ ``$NP/CLP$ equivalence constant", and denote it by $p^*(A,b)$.

Obviously, it is important to evaluate the $NP/CLP$ equivalence constant $p^*(A,b)$ for us to choose an appropriate model $(P_p)$ substituting for $(P_0)$. However, this is a hard and difficult work, although the inequality (\ref{equ15}) can be used to derive a rudimentary estimation. With the relationship (\ref{equ2}), perhaps someone believes that the constant $p(A,b)$ is determined by some single value (that is, if $(P_{p_{1}})$ is equivalent to $(P_0)$, so are all $(P_{p})$ for $p\leq p_1$). Nevertheless, the following example shows that it is not true.

\begin{example}
Consider the minimization $(P_p)$ with respect to the underdetermined linear system $Ax=b$ with
\begin{equation}
A=\left(
  \begin{array}{cccc}
    -\frac{20}{29} & 1 & \frac{31}{87} & 0 \\
    0 & 1 & \frac{8}{15} & 1 \\
    \frac{60}{29} & 0 & \frac{463}{435} & -1 \\
  \end{array}
\right),\ \ b=(1,2,3)^T.
\end{equation}
It is easy to show that the solutions $x=(x_1,x_2,x_3,x_4)^T$ have the following parameterized form
\begin{equation}\label{equ16}
x_1=t, x_2=-\frac{4}{27}+\frac{40}{27}t, x_3=\frac{29}{9}\bigg(1-\frac{20}{29}t\bigg), x_4=\frac{58}{135}\bigg(1-\frac{20}{29}t\bigg)
\end{equation}
where the parameter $t$ varies in $\mathbb{R}$. Hence, it can be the $l_p$-norm is computed by the following function of $t$,
 \begin{equation}\label{equ17}
 \|x\|_p^p=|t|^p+\bigg|-\frac{4}{27}+\frac{40}{27}t\bigg|^p+\bigg|\frac{29}{9}(1-\frac{20}{29}t)\bigg|^p+\bigg|\frac{58}{135}(1-\frac{20}{29}t)\bigg|^p.
 \end{equation}
Obviously, the unique solution to $(P_0)$ is $x_0=(1.45,2,0,0)^T$ (with respect to $t=1.45$), and all the solutions to $(P_p)$ for $0\leq p\leq 1$ exist in the set $B_\infty(2)$. It is also easy to test that the constant $r_m(A,b)$ defined as in (\ref{equ12}) equals to $0.1$. Hence, by the inequality (\ref{equ15}) we can derive that $p^*(A,b)\geq 0.135$. So, by Theorem 1 we know that $x_0=(1.45,2,0,0)^T$ is the unique solution to $(P_p)$ for $0<p<0.135$. From Figure 3 it is seen that the $l_{0.08}$-norm and $l_{0.135}$-norm reach their minimums at $t=1.45$, which corresponds to the unique solution $x_0=(1.45,2,0,0)^T$.

Now we consider three cases when $p=0.8, 0.95$, $1$, respectively. The behaviours of $\|x\|_p^p$ as the functions of $t$ are displayed in Figure 3 for those $p$. By the formula (\ref{equ17}) it is easy to test that $x_{0.8}=(0.1,0,3,0.4)^T$ solve $(P_{0.8})$, and while $x_{0.95}=x_1=(1.45,2,0,0)^T$ solves both $(P_{0.95})$ and $(P_1)$, respectively. But $\|x_{0.8}\|_0=3>\|x_{0.95}\|=2=\|x_0\|_0$. This shows that $p^*(A,b)<0.8$ in spite of the fact that $(P_{0.95})$ possesses the same unique solution as $(P_0)$.
\end{example}

\begin{figure}[h!]
 \centering
  \includegraphics[height=0.5\textwidth,width=0.6\textwidth]{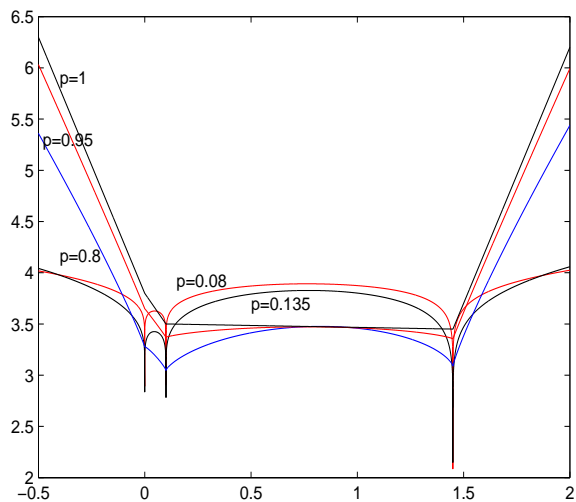}
 \caption{\small{Behaviours of the $l_p$-norms as the functions of $t$ for $p=0.1,0.135,0.8, 0.95$ and $1$, respectively.}}\label{fig3}
\end{figure}

Obviously, the above example also incidentally shows that in the whole interval $(0,1)$ of $p$ it is not true that the smaller $p$ is, the sparser the solutions to $(P_p)$ are.

\section{\textbf{Conclusions}}

Among the numerous substitution models for the $l_0$ minimization problem $(P_0)$, the $l_p$-norm minimizations $(P_p)$ with $0<p\leq 1$ have long been considered as the most natural choices. However, it has not been answered to what extent these models $(P_p)$ can replace $(P_0)$. In this paper, we have exposed the equivalence between $(P_0)$ and $(P_p)$, so completely answered the question mentioned. The established equivalence means that solving $(P_0)$ can be completely attacked by solving the convex minimization $(P_p)$ instead for some small $p$, while the latter is computable by some commonly used means at least for some special $p$. However, it should be pointed out that the main result obtained in this paper is only qualitative, it has not given quantitative characterization to the $NP/CLP$ equivalence constant, which the authors think is important for model choice and then for algorithm design. Altogether, the authors wish this paper can throw away a brick in order to get ages.

\vspace{8mm}
\noindent\textbf{Acknowledgements:}\ \ This work was supported by the NSFC under contact no. 11131006 and EU FP7 Projects EYE2E
(269118) and LIVCODE (295151), and in part by the National Basic Research Program of China under contact no. 2013CB329404.

\end{document}